\documentclass{article}
\usepackage{amscd,amsmath,amssymb,amsfonts,amsthm}
\usepackage{url}
\usepackage{color,verbatim}

\newcommand{\cv}{\mathbf{c}}
\newcommand{\hv}{\mathbf{h}}
\newcommand{\rv}{\mathbf{r}}
\newcommand{\sv}{\mathbf{s}}
\newcommand{\uv}{\mathbf{u}}

\newcommand{\xv}{\mathbf{x}}

\newcommand{\CC}{\mathbb{C}}
\newcommand{\FF}{\mathbb{F}}
\newcommand{\QQ}{\mathbb{Q}}
\newcommand{\RR}{\mathbb{R}}
\newcommand{\ZZ}{\mathbb{Z}}

\newcommand{\tr}{{\rm Tr}}
\newcommand{\diag}{\mathsf{diag}}
\newcommand{\vol}{{\rm vol}}

\newcommand{\pf}{\mathfrak{p}}

\newcommand{\Pf}{\mathfrak{P}}

\newcommand{\Oc}{\mathcal{O}}

\newcommand{\Rc}{\mathcal{R}}

\newtheorem{defn}{Definition}
\newtheorem{exmp}{Example}
\newtheorem{lem}{Lemma}
\newtheorem{rem}{Remark}
\newtheorem{prop}{Proposition}

\newtheorem{cor}{Corollary}

\newcommand\revised[1]{\textcolor{black}{#1}}
\newcommand\WK[1]{\textcolor{black}{#1}}

\title{Construction A of Lattices over Number Fields and Block Fading (Wiretap) Coding}
\author{Wittawat Kositwattanarerk, Soon Sheng Ong, Fr\'ed\'erique Oggier
\thanks{W. Kositwattanarerk is with the Department of Mathematics, 
Faculty of Science, Mahidol University, Thailand. This work was done while he was with 
the Division of Mathematical Sciences, Nanyang Technological University, Singapore. 
S.S Ong and F. Oggier are with the Division of Mathematical Sciences, Nanyang Technological University, Singapore. 
The research of W. Kositwattanarerk (while at Nanyang Technological University), of S.S. Ong and F. Oggier for this work has been supported by the Singapore National
Research Foundation under Research Grant NRF-RF2009-07. 
email: wittawat.kos@mahidol.edu,SSONG1@e.ntu.edu.sg, frederique@ntu.edu.sg. 
Part of this work appeared in the proceedings of the IEEE International Symposium on Information Theory (ISIT 2013), Istanbul, Turkey \cite{isit}.
}
}
\date{\today}

\begin{document}

\maketitle

\begin{abstract}
We propose a lattice construction from totally real and CM fields, which naturally generalizes Construction A of lattices from $p$-ary codes obtained from the cyclotomic field $\QQ(\zeta_p)$, $p$ 
a prime, which in turn contains the so-called Construction A of lattices from binary codes as a particular case.
We focus on the maximal totally real subfield $\QQ(\zeta_{p^r}+\zeta_{p}^{-r})$ of the cyclotomic 
field $\QQ(\zeta_{p^r})$, $r\geq 1$. Our construction has applications to coset encoding of algebraic lattice codes \revised{for block fading channels, and in particular for block fading wiretap channels}.
\end{abstract}

%
%

\section{Introduction}
\label{sec:intro}

Lattices have been extensively studied for addressing the problem of coding for additive white Gaussian noise (AWGN) channels. Lattice codes are proven to be capacity-achieving (e.g.~\cite{UR,ELZ}), and several lattice decoding algorithms are available (e.g.~\cite{VB,SFS}). 
A practical way of encoding lattices is via Construction A \cite{splag,ELZ}, that is, when lattices are constructed from a binary code. Construction A considers a binary code $C$ of length $N$ and the mapping $\rho$ from $\ZZ^N$ to $\FF_2^N$ by reduction modulo 2. Now, the preimage $\rho^{-1}(C)$ (or $\rho^{-1}(C)/\sqrt{2}$) of $C$ forms a lattice in $\RR^N$. An encoding of $\rho^{-1}(C)$ is then easily obtained by using part of the bits to encode the binary code, and the rest to label the lattice $\ZZ^N$.
There have been recent generalizations of constructions of lattices from codes, for example in \cite{bw} where Barnes-Wall lattices are obtained from linear codes over some polynomial rings, resulting in an explicit method of bit-labeling complex Barnes-Wall lattice codes, which in turn allows an efficient decoding of these lattice codes.

Construction A provides a natural way to implement coset encoding for lattices. 
Lattice-type coset codes can be easily implemented and are proven to have an excellent coding gain \cite{F1,F2}. Coset encoding is also useful in the context of Gaussian wiretap lattice codes \cite{osb}. Indeed, the underlying idea of wiretap encoding is to have two nested lattices where the fine lattice is represented as a union of cosets of the coarse one. Information symbols are used to label the cosets, and random bits are introduced to pick a lattice point at random within this coset. This randomized encoding is meant to provide confidentiality between the two legitimate players, in the presence of an eavesdropper.

The theory behind Construction A is also well understood. 
There is a series of dualities between theoretical properties of codes and that of the resulting lattices, linking for example the dual of the code to the dual of the lattice, or the weight enumerator of the code to the theta series of the lattice (see e.g. \cite{splag,Ebeling}). 
Construction A has been generalized in different directions, for example a generalized construction from the cyclotomic field $\QQ(\zeta_p)$, $p$ a prime, is presented in \cite{Ebeling}, while constructions using linear codes over rings have also been considered (e.g.\cite{BDHO}). \revised{There is consequently a rich literature studying Construction A, over different alphabets, from different angles: the modular form one delves into weight enumerators and theta series (e.g. \cite{CDL}), while the lattice point of view aims at getting extremal or modular lattices \cite{Bachoc,BSC}.}

Lattices have also been considered for transmission over fading channels. 
Specifically, algebraic lattices, defined as lattices obtained via the ring of integers of a number field, provide efficient modulation schemes \cite{Oggier-2} for fast Rayleigh fading channels. Families of algebraic lattices are known to reach full diversity, the first design criterion for fading channels. Algebraic lattice codes are then natural candidates for the design of \revised{codes for block fading channels in general, and }
wiretap codes for both fast fading and block fading channels in particular \cite{icc,BO13}: by taking two fully diverse nested lattices built over a number field, reliability for the legitimate user is taken care of, but coset encoding is needed to confuse the eavesdropper, addressing the question of performing coset encoding of lattices built over number fields.

In this paper, we consider a generalized construction of lattices over a number field from linear codes, that generalizes the construction on the cyclotomic field $\QQ(\zeta_p)$, $p$ a prime, exposed in \cite{Ebeling}. \revised{We give a general definition of Construction A in Section \ref{sec:latt}, and quickly focus on number fields which have a prime which totally ramifies, describing discriminant and generator matrix in Subsection \ref{subsec:ram-mat}, and connections between codes and lattices in Subsection \ref{subsec:ram-prop}. Particular constructions are thus obtained from cyclotomic fields and their subfields, as discussed in Subsection \ref{subsec:cyc}.}
 
\revised{
We are furthermore interested in the application of this construction to lattice coset encoding for block fading channels. Totally real subfields of cyclotomic fields are known to provide diversity when used for transmission over fading channels. The proposed Construction A yields a coset encoding scheme for block fading, which inherits the full diversity property from the chosen underlying number field. \WK{The details of this application are discussed in Section \ref{sec:wc}.} As a special case, we describe in Subsection \ref{subsec:wiretap} how the presented construction enables the encoding of wiretap lattice codes over a block fading wiretap channel.
}

%
%
%
\section{Lattices and Codes from Number Fields}
\label{sec:latt}

We first recall the definition of an integral lattice and outline basic results from algebraic number theory.

\begin{defn}\label{def:latt}
An {\it integral lattice} $\Gamma$ is a free $\ZZ$-module of finite rank together with 
a positive definite symmetric bilinear form $\langle ~,~ \rangle: \Gamma \times \Gamma \rightarrow \ZZ$.
\end{defn}

A lattice $\Gamma\subset\RR^m$ is said to have full rank if the rank of $\Gamma$ as a $\ZZ$-module is $m$. We will not consider lattices which are not integral or are not of full rank in what follows.

\begin{defn}\label{def:dual}
Let $\Gamma\subset\RR^m$ be an integral lattice.
The {\it dual lattice} $\Gamma^*$ of $\Gamma$ is 
\[
\Gamma^*=\{y\in\RR^m~|~y \cdot x \in \ZZ \mbox{ for all }x\in\Gamma\},
\]
where $\cdot$ is the usual inner product.
When $\Gamma=\Gamma^*$, the lattice $\Gamma$ is said to be {\it unimodular}.
\end{defn}

\begin{defn}
The discriminant of a lattice $\Gamma$, denoted disc($\Gamma$), is the determinant of $MM^T$ where $M$ is a generator matrix for $\Gamma$. The volume $\vol(\RR^n/\Gamma)$ of a lattice $\Gamma$ is defined to be 
$|\det(M)|$.
\end{defn}

Thus the discriminant is related to the volume of a lattice by 
\[
\vol(\RR^n/\Gamma)=\sqrt{disc(\Gamma)}.
\]
Moreover, when $\Gamma$ is integral, we have disc($\Gamma$) = $|\Gamma^*/\Gamma|$.

We are interested in integral lattices built over number fields, that is, finite field extensions of $\QQ$. Let $K$ be a number field of degree $n$. There exist $n$ embeddings $\sigma_1,\ldots,\sigma_n$ from $K$ into $\CC$. The signature of $K$ is the pair $(r_1,r_2)$ where $r_1$ and $r_2$ are the number of real embeddings and pairs of complex embeddings of $K$ respectively. Note that we have $n=r_1+2r_2$. The field $K$ is said to be totally real if $r_1=n$ and $r_2=0$. A CM-field is a totally imaginary quadratic extension of a totally real number field.
 
The ring of integers $\Oc_K$ of $K$ is a free $\ZZ$-module of rank $n$.
Thus, there exists an integral basis $\{\nu_1,\ldots,\nu_n\}$ of $\Oc_K$ such that every element $w$ of $\Oc_K$ can be written 
as $w=\sum_{i=1}^n w_i \nu_i$ where $w_i \in \ZZ$. Every ideal of $\Oc_K$ has a unique decomposition as a product of prime ideals. 
Let $p$ be a prime in $\ZZ$, then $p\Oc_K=\prod_{i=1}^g \pf_i^{e_i}$
where $\pf_i$ are prime ideals of $\Oc_K$, said to be above $p$. The number $g$ of prime ideals above $p$ and the ramification indices $e_i$ are related to the degree $n$ of the number field $K$ by the formula $n=\sum_{i=1}^g e_if_i$
where $f_i$ is the inertial degree, i.e., the degree of the finite field $\Oc_K/\pf_i$ over $\FF_p$. When $K$ is a Galois extension, this simplifies to $n=efg$.
In other words, $e_i=e$ and $f_i=f$ for all $i$. 

Let $\FF_q$ be the finite field with $q$ elements, where $q$ is a prime power. We use the term code to mean an $(N,k)$ linear code over $\FF_q$, that is, $C$ is a $k$-dimensional subspace of $\FF_q^N$. The dual code $C^\perp$ of $C$ is  given by 
\[
C^\perp=\{ y \in \FF_q^N ~|~ y \cdot x =0 \mbox{ for all }x\in C\}.
\]
\WK{The code $C$ is said to be {\it self-dual} when $C=C^\perp$ and {\it self-orthogonal} when $C\subset C^\perp$.}

\subsection{A General Lattice Construction}

Given a number field $K$ and a prime $\pf\in\Oc_K$ above $p$ where $\Oc_K/\pf\simeq \FF_{p^f}$, let $C$ be an $(N,k)$ linear code over $\FF_{p^f}$. We define $\Gamma_C$ to be the ``preimage'' of $C$ in $\Oc_K^N$. The precise definition is given as follows.
\begin{defn}\label{def:Gamma}
Let $\rho:\Oc_K^N \rightarrow \FF_{p^f}^N$ be the mapping defined by the reduction modulo the ideal 
$\pf$ in each of the $N$ coordinates. Define
\[
\Gamma_C := \rho^{-1}(C) \subset \Oc_K^N.
\]
\end{defn}
To see why $\Gamma_C$ forms a lattice, we first note that $\rho^{-1}(C)$ is a subgroup of $\Oc_K^N$ since $C$ is a subgroup of $\FF_{p^f}^N$. Furthermore, $\Oc_K^N$ is a free $\ZZ$-module of rank $nN$, and so it follows that $\rho^{-1}(C)$ is also a free $\ZZ$-module. Now, since $|\Oc_K^N/\pf^N|<\infty$, $\rho^{-1}(C)$ and $\Oc_K^N$ must have the same rank as a $\ZZ$-module. We conclude that $\rho^{-1}(C)$ is a $\ZZ$-module of rank $nN$.

There are variations of the bilinear form $\langle ~,~ \rangle: \Gamma \times \Gamma \rightarrow \ZZ$ depending on the field $K$ and the code $C$. We now focus on the case when $K$ is totally real. In other words, the signature of $K$ is $(n,0)$. Let $x=(x_1,\ldots,x_N)$ and $y=(y_1,\ldots,y_N)$ be vectors in $\Oc_K^N$. Then, $\rho^{-1}(C)$ forms a lattice with the symmetric bilinear form 
\begin{equation}\label{eq:trtr}
\langle x,y \rangle = \sum_{i=1}^N \tr_{K/\QQ}(\alpha x_i y_i),
\end{equation}
where $\alpha\in \Oc_K$ is totally positive, meaning that $\sigma_i(\alpha)>0$ for all $i$. The totally positive condition ensures that the trace form is positive definite; if $\sigma_i(\alpha)>0$ for all $i$ and $x$ is not the zero vector, then, by the definition of trace,
\begin{align*}
\langle x,x \rangle &= \sum_{i=1}^N \tr_{K/\QQ}(\alpha x_i x_i) \\
                    &= \sum_{i=1}^N \sum_{j=1}^n \sigma_j(\alpha)\sigma_j(x_i)^2 >0.
\end{align*}
If $\alpha\in\Oc_K$ then $\tr_{K/\QQ}(\alpha x_i y_i)$ belongs to $\ZZ$ for $i$ since $x_i,y_i \in \Oc_K$. Consequently, $\langle x,y \rangle$ is an integer. Thus, a totally positive $\alpha\in\Oc_K$ guarantees that $\Gamma_C$ together with the bilinear form \eqref{eq:trtr} is an integral lattice. Though, as will be shown later, depending on the code $C$, other choices of $\alpha$ might be possible.

Variations of the above construction have been considered in the literature. When $N=1$ (i.e., codes are not involved) the problem reduces to understanding which lattices can be obtained on the ring of integers of a number field, see \cite{eva} for example. \WK{In \cite{Ebeling}, Ebeling  considers the construction when $K$ is a cyclotomic field $\QQ(\zeta_p)$, and we will outline this construction in Example \ref{ex:Ebeling}. Its maximal real subfield $\QQ(\zeta_p+\zeta_p^{-1})$ is studied in \cite{CDL} and \cite{isit}. In \cite{CDL}, the reduction is done by the ideal $(2m)$, yielding codes over a ring of polynomials with coefficients modulo $2m$ whereas in \cite{isit} the reduction is done by the ideal $(2-\zeta_p+\zeta_p^{-1})$ and the resulting codes are over $\FF_p$. Alternatively, quadratic extensions $K=\QQ(\sqrt{-l})$ are considered for example in \cite{Bachoc,DKL} where the reduction is done by the ideal $(p^e)$ and the resulting codes are over the ring $\Oc_K/p^e\Oc_K$. In general, if $K$ is a number field and $\pf$ is a prime above $p$ with a large ramification index, then it is known \cite{itw} that the quotient $\Oc_K/p\Oc_K$ is a polynomial ring, and the ideal $\pf\Oc_K/p\Oc_K$ corresponds to one of its ideals, which in turn defines a code over the given polynomial ring. In the next subsection, we will look into the case when $\pf$ is totally ramified.}

\subsection{The Case of a Totally Ramified Prime: Discriminant and Generator Matrix}
\label{subsec:ram-mat}

We now focus on the case where $K$ is a Galois extension and the prime $\pf$ is chosen so that $\pf$ is totally ramified. Therefore, we have $p\Oc_K=\pf^n$, $e=n$, and $f=1$.

Now, let $C \subset \FF_p^N$ be a linear code over $\FF_p$ of length $N$, and 
let $\rho:\Oc_K^N \rightarrow \FF_p^N$ be the mapping defined by the reduction modulo the prime ideal $\pf$ on every coordinate. We consider (see Definition \ref{def:Gamma}) the lattice
\[
\Gamma_C := \rho^{-1}(C) \subset \Oc_K^N.
\]
We know from the previous subsection that $\Gamma_C$ is a lattice of rank $nN$. 

A generator matrix for the lattice $\Gamma_C=\rho^{-1}(C)$ is computed next. Recall that each of the $N$ coordinates of a lattice point $x=(x_1,\ldots,x_N)\in\Gamma_C$ is an element of $\Oc_K$ since $\Gamma_C \subset \Oc_K^N$. Here, $\Gamma_C$ has rank $nN$ as a free $\ZZ$-module, so we are interested in a $\ZZ$-basis of $\Gamma_C$. Let $\{\nu_1,\ldots,\nu_n\}$ be a $\ZZ$-basis of $\Oc_K$. Then, a generator matrix for the lattice formed by $\Oc_K$ together with the standard trace form $\langle w,z \rangle =\tr_{K/\QQ}(wz)$, $w,z\in\Oc_K$, is given by
\begin{equation}\label{eq:M}
M:=
\begin{pmatrix}
\sigma_1(\nu_1) & \sigma_2(\nu_1) & \ldots & \sigma_n(\nu_1) \\
\vdots          & \vdots          &        & \vdots \\
\sigma_1(\nu_n) & \sigma_2(\nu_n) & \ldots & \sigma_n(\nu_n) 
\end{pmatrix}
\end{equation}
since $MM^T=\tr_{K/\QQ}(\nu_i\nu_j)$. A vector $w$ in this lattice is thus a linear combination  $w=\sum_{i=1}^n w_i \nu_i$ of the rows of $M$ which is then embedded into $\RR^n$ as $(\sigma_1(\sum_{j=1}^nw_j\nu_j),\ldots,\sigma_n(\sum_{i=j}^nu_j\nu_j))$, and 
\[
\langle w,z \rangle = \tr_{K/\QQ}(wz)
\]
as it should be (this is the bilinear form (\ref{eq:trtr}) when $N=1$ and $\alpha=1$). 

We now give one last ingredient before we derive a generator matrix for the lattice $\Gamma_C$. Recall that the prime ideal $\pf$ is a $\ZZ$-module of rank $n$. It then has a $\ZZ$-basis $\{\mu_1,\ldots,\mu_n\}$ where $\mu_i=\sum_{j=1}^n \mu_{ij}\nu_j$, $\mu_{ij}\in\ZZ$, so that
\[
\begin{pmatrix}
\sigma_1(\mu_1) & \ldots & \sigma_n(\mu_1)\\
\vdots & & \vdots \\
 \sigma_1(\mu_n) & \ldots & \sigma_n(\mu_n) \\
\end{pmatrix}
=
\begin{pmatrix}
\sum_{j=1}^n \mu_{1j}\sigma_1(\nu_j) & \ldots &  \sum_{j=1}^n \mu_{1j}\sigma_n(\nu_j)\\
\vdots & & \vdots \\
\sum_{j=1}^n \mu_{nj}\sigma_1(\nu_j) & \ldots & \sum_{j=1}^n \mu_{nj}\sigma_n(\nu_j) \\
\end{pmatrix}
=
DM
\]
where $D:=(\mu_{ij})_{i,j=1}^n$.

\begin{prop}\label{prop:MC}
The lattice $\Gamma_C$ is a sublattice of $\Oc_K^N$ with discriminant 
\[
disc(\Gamma_C)=d_K^{N}(p^f)^{2(N-k)}
\]
where $d_K=(\det(\sigma_i(w_j))_{i,j=1}^n)^2$ is the discriminant of $K$.
The lattice $\Gamma_C$ is given by the generator matrix
\[
M_C=
\begin{pmatrix}
I_k \otimes M  & A \otimes M \\
{\bf 0}_{n(N-k),nk} & I_{N-k}\otimes DM \\
\end{pmatrix}
\]
where $\otimes$ is the tensor (also known as Kronecker) product of matrices, 
$(I_k~~A)$ is a generator matrix of $C$, $M$ is the matrix of embeddings of a $\ZZ$-basis of $\Oc_K$ given in (\ref{eq:M}), and $DM$ is the matrix of embeddings of a $\ZZ$-basis of $\pf$.
\end{prop}

\begin{proof}
The bilinear form $\langle w,z \rangle =\tr_{K/\QQ}(wz)$, $w,z\in\Oc_K$, has determinant $d_K$ over $\Oc_K$ since $d_K=\det(MM^T)$ by definition. Thus, the bilinear form $\langle x,y \rangle = \sum_{i=1}^N\tr(x_iy_i)$ has determinant $d_K^N$ over $\Oc_K^N$.
The map $\rho$ of reduction$\mod \pf$ is surjective, and $\rho^{-1}(C)$ is of index $(p^f)^{N-k}$,  therefore the discriminant of $\Gamma_C$ is
\[
disc(\Gamma_C)=d_K^N(p^f)^{2N-2k}.
\]
It is clear from the shape of the generator matrix $M_C$ that this lattice has the right rank. Note that the first $nk$ row of $M_C$ correspond to an embedding of a basis for $C$ and the last $n(N-k)$ rows of $M_C$ correspond to an embedding of a basis of $\pf$. To make this more precise, let us write $u_i=(u_{i1},\ldots,u_{in})\in\ZZ^n$ where $x_i=\sum_{l=1}^nu_{il}\nu_l$, $i=1,\ldots,N$, and define $\sigma=(\sigma_1,\ldots,\sigma_n):\Oc_K \rightarrow \RR^{n}$ to be the canonical embedding of $K$. We have
\[
\sigma_j(x_i)=\sigma_j\left(\sum_{l=1}^nu_{il}\nu_l\right)=u_i \cdot (M_{lj})_{l=1}^n
\]
and
\begin{eqnarray*}
&&
(u_1,\ldots,u_k,u_{k+1},\ldots,u_{N})
\begin{pmatrix}
I_k\otimes M & A \otimes M \\
{\bf 0}_{n(N-k),nk} & I_{N-k}\otimes DM \\
\end{pmatrix}\\
&=&
(\sigma(x_1),\ldots,\sigma(x_k),\sum_{j=1}^ka_{j,1}\sigma(x_j)+\sigma(x'_{k+1}),\ldots,
\sum_{j=1}^ka_{j,N-k}\sigma(x_j)+\sigma(x'_N))
\end{eqnarray*}
where $x'_{k+1},\ldots,x'_N$ are in the ideal $\pf$.
It is not hard to see that the above vector is an element in $\Gamma_C$. Indeed, if we define
\[
\psi:\sigma(x_i)=(\sigma_1(x_i),\ldots,\sigma_n(x_i)) \mapsto x_i=\sum_{l=1}^nu_{il}\nu_l \in \Oc_K,
\]
then applying $\psi$ and $\rho$ componentwise in order gives
\[
c=(\rho(\psi(\sigma(x_1))),\ldots,\rho(\psi(\sigma(x_k))),\sum_{j=1}^ka_{j,1}\rho(\psi(\sigma(x_j))),\ldots,
\sum_{j=1}^ka_{j,N-k}\rho(\psi(\sigma(x_j)))),
\]
since $x'_i$ reduces to zero modulo $\pf$. Now, $c$ is a codeword of the code $C$ given by
\[
c=(\rho(\psi(\sigma(x_1))),\ldots,\rho(\psi(\sigma(x_k)))\cdot (I_k~~A).
\]

Finally, the absolute value of the determinant of $M_C$ can be computed as
\begin{align*}
|\det(M_C)| &= |\det(I_k \otimes M)\det(I_{N-k}\otimes DM)| \\
            &= |\det(M)|^k|\det(DM)|^{N-k} \\
            &= |\det(M)|^N |\det(D)|^{N-k}\\
            &= \sqrt{d_K}^NN(\pf)^{N-k}\\
            &=\sqrt{d_K}^N (p^f)^{N-k}, 
\end{align*}
showing that $M_C$ generates a lattice with the same volume as $\Gamma_C$, which completes the proof of the proposition.
\end{proof}

\WK{We remark that the advantage of considering the case of a prime which is totally ramified is that the matrix $A$ can be easily lifted, since it has coefficients in $\FF_p$. 
}

Now, for $x=(x_1,\ldots,x_N)\in \Gamma_C \subset \Oc_K^N$, we have that $x_i=\sum_{j=1}^nx_{ij}\nu_j$ for $i=1,\ldots,k$, and the above result also tells us that $x$ is embedded into $\RR^{nN}$ as
\begin{eqnarray}
x&\!\!=\!\!&(\sigma(x_1),\ldots,\sigma(x_k),\sum_{j=1}^ka_{j,1}\sigma(x_j)+\sigma(x'_{k+1}),\ldots,
\sum_{j=1}^ka_{j,N-k}\sigma(x_j)+\sigma(x'_N)) \nonumber \\
&\!\!=\!\!&(\sigma_1(x_1),\ldots,\sigma_n(x_1),\ldots,\sigma_1(x_N),\ldots,\sigma_n(x_N)) \label{eq:x}
\end{eqnarray}
where $x_{k+1}=\sum_{j=1}^ka_{j,1}x_j+x'_{k+1},\ldots,x_N=\sum_{j=1}^ka_{j,N-k}x_j+x'_N$.
Then,
\[
\langle x,y \rangle = \sum_{i=1}^N\tr_{K/\QQ}(x_iy_i).
\]

\begin{cor}\label{cor:disc}
The lattice $\Gamma_C$ has Gram matrix
\[
\begin{pmatrix}
GG^T\otimes \tr(\nu_i\nu_j) & A \otimes \tr(\mu_i\nu_j) \\
A^T \otimes \tr(\mu_i\nu_j) & I_{N-k}\otimes \tr(\mu_i\mu_j) \\
\end{pmatrix}
\]
where $\mu_1,\ldots,\mu_n$ is a $\ZZ$-basis of $\pf$, and $G=(I_k~~A)$.
In particular, $\Gamma_C$ is an integral lattice.
\end{cor}
\begin{proof}
By definition, the Gram matrix of $\Gamma_C$ is
\begin{eqnarray*}
M_CM_C^T &\!\!=\!\!&
\begin{pmatrix}
I_k\otimes M & A \otimes M \\
{\bf 0}_{n(N-k),nk} & I_{N-k}\otimes DM \\
\end{pmatrix}
\begin{pmatrix}
I_k\otimes M^T & {\bf 0}_{nk,n(N-k)}  \\
A^T \otimes M^T & I_{N-k}\otimes M^TD^T \\
\end{pmatrix}\\
&\!\!=\!\!&
\begin{pmatrix}
(I_k+AA^T)\otimes MM^T & A \otimes MM^TD^T \\
A^T \otimes DMM^T & I_{N-k}\otimes DMM^TD^T \\
\end{pmatrix}\\
&\!\!=\!\!&
\begin{pmatrix}
GG^T\otimes \tr(\nu_i\nu_j) & A \otimes \tr(\mu_i\nu_j) \\
A^T \otimes \tr(\mu_i\nu_j) & I_{N-k}\otimes \tr(\mu_i\mu_j) \\
\end{pmatrix}.
\end{eqnarray*}
It follows that the entries of this matrix are integers, and $\Gamma_C$ is an integral lattice.
\end{proof}

A similar construction is obtained from a CM-field. We provide an outline next, \revised{with fewer} details, since we will be mostly interested in the totally real case. Recall that a CM-field is a totally imaginary quadratic extension of a totally real number field. If $K$ is a CM-field and $\alpha\in \Oc_K \cap \RR$ is totally positive, then $\rho^{-1}(C)$ forms a lattice with the symmetric bilinear form 
\begin{equation}\label{eq:trcm}
\langle x,y \rangle = \sum_{i=1}^N \tr_{K/\QQ}(\alpha x_i \bar{y}_i),
\end{equation}
where $\bar{y}_i$ denotes the complex conjugate of $y_i$. 
Since $\sigma_i$ commutes with the complex conjugation and $\sigma_i(\alpha)>0$ for all $i$, 
we have for $x$ not the zero vector that
\begin{align*}
\langle x,x \rangle &= \sum_{i=1}^N \tr_{K/\QQ}(\alpha x_i \bar{x}_i) \\
                    &= \sum_{i=1}^N \sum_{j=1}^n \sigma_j(\alpha)|\sigma_j(x_i)|^2 >0. 
\end{align*}
Again, whether $\tr_{K/\QQ}(\alpha x_i \bar{y}_i) \in \ZZ$ depends on the choice of $\alpha$, and $\alpha\in\Oc_K$ is a sufficient condition, as that makes  $\alpha x_i \bar{y}_i \in \Oc_K$ for all $x_i,y_i\in\Oc_K$. A generator matrix for this lattice is obtained similarly as above, with the exception that now $\sigma_{r_1+1},\ldots,\sigma_{n}$ are complex embeddings. One thus chooses one complex embedding per pair and separates its real and imaginary parts to get \[\sigma=(\sigma_1,\ldots,\sigma_{r_1},\Re(\sigma_{r_1+1}),\Im(\sigma_{r_1+1}),\ldots,\Re(\sigma_{r_1+r_2})\Im(\sigma_{r_1+r_2}))\]
and
\[
x=(\sigma_1(\sum_{j=1}^n u_j\nu_j),\ldots,\Re(\sigma_{r_1+1}(\sum_{j=1}^n u_j\nu_j)),
\ldots, \Im(\sigma_{r_1+r_2}(\sum_{j=1}^n u_j\nu_j))).
\]

\subsection{\WK{The Case of a Totally Ramified Prime and Self-Orthogonal Codes}}
\label{subsec:ram-prop}

We know from the previous subsection that $\Gamma_C$ is an integral lattice of rank $nN$ with respect to the bilinear form $\langle x,y \rangle = \sum_{i=1}^N \tr_{K/\QQ}(\alpha x_i \bar{y}_i)$ for a totally positive element $\alpha\in\Oc_K\cap \RR$. If $K$ is totally real, then $\bar{y}_i=y_i$, and this notation allows us to treat both the cases of totally real and CM-fields at the same time. \WK{In this subsection, we derive some properties of the lattice
\[
\Gamma_C := \rho^{-1}(C) \subset \Oc_K^N
\]
when $C$ is self-orthogonal, i.e. $C\subset C^\perp$.}

Let $C$ be an $(N,k)$ linear code over $\FF_p$. We will show next that if $C\subset C^\perp$, 
then $\sum_{i=1}^N \tr_{K/\QQ}(x_i \bar{y}_i)\in p\ZZ$, and thus we can normalize the symmetric bilinear form by a factor of $1/p$, or equivalently, choose $\alpha=1/p$.

\begin{lem}\label{lem:int}
Let $C\subset \FF_p^N$ be a code such that $C\subset C^\perp$. Then, $\Gamma_C$ is an integral lattice with respect to the bilinear form $\langle x,y \rangle = \sum_{i=1}^N \tr_{K/\QQ}(x_i \bar{y}_i/p)$. 
\end{lem}

\begin{proof}
It suffices to show that $\langle x,y \rangle$ as defined is an integer for all $x,y\in\Gamma_C$. Let $x=(x_1,\ldots,x_N)$ and $y=(y_1,\ldots,y_N)$ be elements in $\Gamma_C=\rho^{-1}(C)$. It is not hard to see that
\begin{align*}
\rho(x\cdot y) & =\rho\left(\sum_{i=1}^N x_iy_i\right) \\
& =\sum_{i=1}^N \rho(x_i)\rho(y_i) \\
& =\rho(x)\cdot \rho(y) \\
&= 0\in\FF_p
\end{align*}
where the last equality holds since $\rho(x),\rho(y)\in C$ and $C\subset C^\perp$. It follows that
\[
\sum_{i=1}^Nx_iy_i = x\cdot y \equiv 0 \pmod \pf.
\]

We are going to show next that $\bar{y_i}\equiv y_i \pmod \pf$ for all $i=1,\ldots,N$. Since $\Oc_K/\pf\cong\FF_p$, for each $i$, one can write $y_i\in\Oc_K$ as $y_i=y'_i+y''_i$ where $y'_i\in\ZZ$ and $y''_i\in \pf$. Note that $\bar{\cdot}$ is the automorphism of $K$ induced by complex conjugation. Since $K$ is a Galois extension, the Galois group acts transitively on the ideals above $p$. However, the only prime above $p$ is $\pf$, so we must have $\bar{y}''_i\in \pf$. It follows that $\bar{y_i}=y'_i+\bar{y}''_i\equiv y'_i+y''_i =y_i \mod \pf$ as desired.

Therefore,
\[
\sum_{i=1}^Nx_iy_i \equiv \sum_{i=1}^Nx_i\bar{y_i} \equiv 0 \pmod \pf.
\]
Again, as $\pf$ is the only prime above $p$, all conjugates of $\sum_{i=1}^Nx_i\bar{y}_i$ must lie in $\pf$, and so must its trace. In other words,
\[\tr_{K/\QQ}\left( \sum_{i=1}^N x_i \bar{y}_i\right)\in\pf,\]
implying that
\[\tr_{K/\QQ}\left( \sum_{i=1}^N x_i \bar{y}_i\right)\in\pf\cap \ZZ= p\ZZ.\]
Now, by linearity of the trace,
\[
\langle x,y \rangle  =  \sum_{i=1}^N \tr_{K/\QQ}(x_i \bar{y}_i/p) 
                     =  \frac{1}{p} \tr_{K/\QQ}\left( \sum_{i=1}^N x_i \bar{y}_i\right),
\]
and so we may conclude that $\Gamma_C$ is integral.
\comment{It now suffices to show that $\tr_{K/\QQ}\left( \sum_{i=1}^N x_i \bar{y}_i\right) \in p\ZZ$.  We know that $\sum_{i=1}^Nx_i\bar{y}_i \in \pf$ and again $\pf$ is the only prime above $p$. Therefore,  This concludes the proof that $\Gamma_C$ is integral, since $\tr_{K/\QQ}\left( \sum_{i=1}^N x_i \bar{y}_i\right) \in \pf \cap \ZZ= p\ZZ$.}
\end{proof}
\begin{rem}\label{rem:norm}\rm
Note that instead of considering the lattice $\rho^{-1}(C)$ with $\langle x,y \rangle = \sum_{i=1}^N\tr_{K/\QQ}(x_i\bar{y}_i/p)$, we can alternatively consider the lattice $\rho^{-1}(C)/\sqrt{p}$ with $\langle x,y \rangle = \sum_{i=1}^N\tr_{K/\QQ}(x_i\bar{y}_i)$. 
\end{rem}

Next, we derive some parameters of $\Gamma_C$ where the code $C$ is chosen so that $C\subset C^\perp$. Let $G=(I_k~~A)$ be a generator matrix for $C$. First, by Proposition \ref{prop:MC}, the generator matrix for $\Gamma_C$ is
\[
M_C=\frac{1}{\sqrt{p}}
\begin{pmatrix}
I_k \otimes M  & A \otimes M \\
{\bf 0}_{n(N-k),nk} & I_{N-k}\otimes DM \\
\end{pmatrix}.
\]
Its discriminant is then
\[
disc(\Gamma_C)=d_K^Np^{2N-2k-nN},
\]
as may be computed directly from the determinant of $M_C$, or by using Corollary \ref{cor:disc}:
\[
disc(\Gamma_C)=(\frac{1}{p})^{nN}d_K^N(p^f)^{2(N-k)}=d_K^N \frac{p^{2(N-k)}}{p^{nN}}.
\]
The Gram matrix of $\Gamma_C$ is
\[
\frac{1}{p}
\begin{pmatrix}
GG^T\otimes \tr(\nu_i\nu_j) & A \otimes \tr(\mu_i\nu_j) \\
A^T \otimes \tr(\mu_i\nu_j) & I_{N-k}\otimes \tr(\mu_i\mu_j) \\
\end{pmatrix}.
\]
We know from Lemma \ref{lem:int} that $\Gamma_C$ is an integral lattice, so we may expect entries of the Gram matrix of $\Gamma_C$ to be integers. This is certainly the case since $G$ is a generator matrix for a code $C\in C^\perp$ over $\FF_p$ (so $p$ divides every entry of $GG^T$) and $\mu_i\nu_j,\mu_i\mu_j\in\pf$ for all $i$ and $j$ (so $p$ divides every entry of $\tr(\mu_i\nu_j)$ and $\tr(\mu_i\nu_j)$).

Several particular cases of the above constructions have been considered in the literature. We provide some of them here as the following examples.

\begin{exmp}\label{ex:Ebeling}\rm
The following construction is discussed in Section 5.2 of \cite{Ebeling}.
Let $p$ be an odd prime, and let $\zeta_p$ be the primitive $p^{\mathrm{th}}$ root of unity. 
Consider the cyclotomic field $K=\QQ(\zeta_p)$  
with the ring of integers $\Oc_K=\ZZ[\zeta_p]$. The degree of $K$ over $\QQ$ is $p-1$, and $p$ is totally ramified, with $p\Oc_K=(1-\zeta_p)^{p-1}$.
Thus, take the prime ideal $\pf=(1-\zeta_p)$ with the residue field $\Oc_K/\pf \simeq \FF_p$ and the 
bilinear form $\langle x,y \rangle = \sum_{i=1}^N \tr_{K/\QQ}(x_i \bar{y_i}/p)$.   
Since $\QQ(\zeta_p)$ is a CM-field, this bilinear form corresponds to (\ref{eq:trcm}) with $\alpha=1/p$.
It was proven that, given a code $C$ over $\mathbb{F}_p$, if $C\subset C^\perp$ then $\rho^{-1}(C)$ is an integral lattice of rank $N(p-1)$. 
\end{exmp}
\begin{exmp}\label{ex:a}\rm
A particularly well-known construction of lattices from codes is when $p=2$ in the previous example. In such case, $\zeta_p=-1$, $\Oc_K=\ZZ$, and $\pf=2\ZZ$, yielding the so-called Construction A (see Section 1.3 of \cite{Ebeling}). To obtain lattices of rank $N$ from binary linear codes of length $N$, we consider
\[
\Gamma_C = \frac{1}{\sqrt{2}}\rho^{-1}(C),
\]
with $\langle x,y \rangle = \sum_{i=1}^N\tr(x_iy_i)$ (see Remark \ref{rem:norm}).
A generator matrix for this lattice is
\[
M_C=
\frac{1}{\sqrt{2}}
\begin{pmatrix}
I_k & A \\
0 & 2I_{N-k} \\
\end{pmatrix},
\]
which may be obtained as a particular case of Proposition \ref{prop:MC}.
\end{exmp}

\subsection{Maximal Totally Real Subfields of Cyclotomic Fields}
\label{subsec:cyc}

Since we consider the case where $\pf$ is a prime above $p$ which totally ramifies in $K$, 
cyclotomic fields and their subfields are natural candidates to study. 
Let $p$ be an odd prime, and let $\zeta_{p^r}$ be a primitive $p^r$th root of unity.

Let us start by considering $K^+=\QQ(\zeta_{p^r}+\zeta_{p^r}^{-1})$, $r\geq1$, the maximal totally real subfield of the cyclotomic field $K=\QQ(\zeta_{p^r})$, with respective rings of integers $\Oc_{K^+}=\ZZ[\zeta_{p^r}+\zeta_{p^r}^{-1}]$ and $\Oc_K=\ZZ[\zeta_{p^r}]$ \cite[p.16]{CF1}. The degree of $K^+$ over $\QQ$ is $\frac{p^{r-1}(p-1)}{2}$. 
The prime $p$ totally ramifies in $K$:
\[
p\Oc_K=\Pf^{p^{r-1}(p-1)},
\]
where $\Pf$ is a principal prime ideal with generator $1-\zeta_{p^r}$
and residue field  $\Oc_K/\pf \simeq \FF_p$. 
We write $e(\Pf|p)=p^{r-1}(p-1)$ for its ramification index, and by transitivity of 
ramification indices
\[
p^{r-1}(p-1)=e(\Pf|p)=e(\Pf|\pf)e(\pf|p)
\]
for $\pf$ the prime above $p$ in $K^+$. This is sufficient to conclude that $e(\pf|p)=p^{r-1}(p-1)/2$, and the prime $p$ also totally ramifies in $K^+=\QQ(\zeta_{p^r}+\zeta_{p^r}^{-1})$:
\[
p\Oc_{K^+}=\pf^{\frac{p^{r-1}(p-1)}{2}}
\]
with
\[
\pf = \Pf \cap \Oc_{K^+}=((1-\zeta_{p^r})(1-\zeta_{p^r}^{-1}))=(2-\zeta_{p^r}-\zeta_{p^{-r}}).
\]

\begin{lem}
Consider the number field $K=\QQ(\zeta_{p^r}+\zeta_{p^r}^{-1})$.
Let $C\subset\mathbb{F}_p^N$ be a $k$-dimensional code such that $C\subset{}C^\perp$. Then the lattice $\Gamma_C$ given in Definition \ref{def:Gamma} together with the bilinear form $\langle x,y \rangle = \sum_{i=1}^N \tr_{K/\mathbb{Q}}(\alpha x_i y_i)$ is an integral lattice of rank $Np^{r-1}(p-1)/2$. 
\end{lem}
This follows immediately from what was done in the previous section.
A generator matrix for the lattice $\Gamma_C=\rho^{-1}(C)$ is obtained as described in Proposition \ref{prop:MC}, namely 
\[
M_C=\frac{1}{\sqrt{p}}
\begin{pmatrix}
I_k \otimes M  & A \otimes M \\
{\bf 0}_{n(N-k),nk} & I_{N-k}\otimes DM \\
\end{pmatrix}
\]
where, as usual, $G=(I_k~~A)$ is a generator matrix of $C$.
Here $\pf$ is principal, and generated by 
$(2-\zeta_{p^r}-\zeta_{p^r}^{-1})$. A $\ZZ$-basis of $\Oc_{K^+}=\ZZ[\zeta_{p^r}+\zeta_{p^r}^{-1}]$ is $\{\zeta_{p^r}^i+\zeta_{p^r}^{-i}\}_{i=0}^{n-1}$. The matrix $M$ is thus obtained from this $\ZZ$-basis by applying the $n$ embeddings of $K$, which are of the form $\zeta_{p^r}+\zeta_{p^r}^{-1}\mapsto\zeta_{p^r}^i+\zeta_{p^r}^{-i}$, with $i$ coprime to $p$.
 
\revised{The case $r=1$, corresponding to $K^+=\QQ(\zeta_p+\zeta_p^{-1})$, has been treated in \cite{isit}.}

\begin{lem}
Let $K^+=\QQ(\zeta_p+\zeta_p^{-1})$, and let
$C\subset \FF_p^N$ be a $k$-dimensional code such that $C\subset C^\perp$. Then
\[
\Gamma_C^* = \Gamma_{C^\perp}.
\]
\end{lem}
\begin{proof}
Let $x\in \Gamma_C$, $y\in \Gamma_{C^\perp}$. Then by definition of these lattices, 
$\rho(x)\in C$ and $\rho(y)\in C^\perp$, and it follows by definition of $C^\perp$ that 
$\rho(x)\cdot \rho(y) \equiv 0 \pmod p$. By an argument similar to that in the proof of Lemma \ref{lem:int}, we deduce that $\langle x,y \rangle \in \ZZ$, and thus $\Gamma_{C^\perp}\subset \Gamma_C^*$.

The discriminant of $\Gamma_C$ is 
\[
disc(\Gamma_C)=p^{N-2k}.
\]
This follows from the fact that
\[
disc(\Gamma_C)=d_{K^+}^N p^{2N-2k-N(p-1)/2}=(p^{(p-1)/2 -1})^N p^{2N-2k-N(p-1)/2},
\]
since $d_{K^+}=p^{(p-1)/2 -1}$.

It then follows that 
\[
\vol(\mathbb{R}^{nN}/\Gamma_C)=(p^{N-2k})^{1/2}=p^{\frac{N}{2}-k}
\]
and
\[
\vol(\mathbb{R}^{nN}/\Gamma_C^*)=p^{k-\frac{N}{2}}.
\]
On the other hand, the dimension of $C^\perp$ is $N-k$, and so
\[disc(\Gamma_{C^\perp})=p^{N-2(N-k)}=p^{2k-N},\]
implying that 
\[\vol(\mathbb{R}^{nN}/\Gamma_{C^\perp})=p^{k-\frac{N}{2}}.\]
We may now conclude that $\Gamma_C^*=\Gamma_{C^\perp}$.
\end{proof}

\begin{cor}
Let $K=\QQ(\zeta_p+\zeta_p^{-1})$ and let $C\subset\mathbb{F}_p^N$ be a $k$-dimensional code such that $C\subset{}C^\perp$. Then the lattice $\Gamma_C$ given in Definition \ref{def:Gamma} together with the bilinear form $\langle x,y \rangle = \sum_{i=1}^N \tr_{K/\mathbb{Q}}(\alpha x_i y_i)$ is an integral lattice of rank $Np^{r-1}(p-1)/2$. In addition, if $C$ is self-dual, then $\Gamma_C$ is an odd unimodular lattice.
\end{cor}
\begin{proof}
By an odd integral lattice, we mean an integral lattice $\Gamma$ which contains a vector $x\in\Gamma$ such that $\langle x,x \rangle$ is an odd integer.
Indeed, take $x=(2-\zeta_{p^r}-\zeta_{p^r}^{-1},0,\ldots,0) \in \Gamma$.
Then
\begin{align*}
\langle x,x\rangle & = \tr_{K^+/\QQ}((2-\zeta_{p}-\zeta_{p}^{-1})^2/p)\\
                   & = \frac{1}{p}\tr_{K^+/\QQ}(6-4(\zeta_{p}+\zeta_{p}^{-1})+(\zeta_{p}^2+\zeta_{p}^{-2}))\\
                   & = \frac{6(p-1)}{2p}+\frac{-3}{p}\tr_{K^+/\QQ}(\zeta_{p}+\zeta_{p}^{-1})\\
&= \frac{6(p-1)}{2p}+\frac{3}{p}=3
\end{align*}
since $\zeta_{p}+\zeta_{p}^{-1}$ and $\zeta_{p}^2+\zeta_{p}^{-2}$ are conjugate, 
and using the fact that
\begin{align*}
\tr_{K^+/\QQ}(\zeta_{p}+\zeta_{p}^{-1})
&=\tr_{K^+/\QQ}(\tr_{\QQ(\zeta_{p})/K^+}(\zeta_{p}))\\
&=\tr_{\QQ(\zeta_{p})/\QQ}(\zeta_{p})=-1.
\end{align*}
\end{proof}

\begin{exmp}\rm
Fix $p=5$, $K=\QQ(\zeta_5)$ and $K^+=\QQ(\zeta_5+\zeta_5^{-1})$, with 
$\rho:\ZZ[\zeta_5+\zeta_5^{-1}]^2 \rightarrow \FF_5^2$. 
Let $\xi=\zeta_5+\zeta_5^{-1}$. The degree of $K^+/\QQ$ is 2, and the two embeddings of $K$ are 
$\sigma_1$ which is the identity and $\sigma_2$ which maps $\zeta_5+\zeta_5^{-1}$ to 
$\zeta_5^2+\zeta_5^{-2}$, that is $\sigma_2(\xi)=-1-\xi$.

Consider the self-dual code of length 2 over $\mathbb{F}_5$ given by a generator matrix
\[
G=
\left(\begin{array}{cc}
1 & 2
\end{array}
\right),
\]
that is, $C=\{(0,0),(1,2),(2,4),(3,1),(4,3)\}$. Then, $\Gamma_C$ is an odd unimodular lattice of rank 4, and the only such lattice is $\mathbb{Z}^4$.


We next compute a generator matrix for the lattice $\Gamma_C$ explicitly. 
We choose the basis $\{1,\xi\}$ for $\Oc_K$, and it follows that the generator matrix for the lattice $\Oc_K$ together with the trace form $\langle x,y \rangle = \tr_{K/\QQ}(xy/5)$, $x,y\in\Oc_K$, is
\[
M=\frac{1}{\sqrt{5}}
\begin{pmatrix}
1 & 1 \\
\xi & \sigma_2(\xi)
\end{pmatrix}
=
\frac{1}{\sqrt{5}}
\left(\begin{array}{cc}
1 & 1 \\
\xi & -1-\xi
\end{array}
\right).
\]
The generator matrix for $\Gamma_C$ as a free $\ZZ$-module of rank $4$ is
\[
M_C=
\begin{pmatrix}
M & 2M \\
{\bf 0} & DM \\
\end{pmatrix}
=
\frac{1}{\sqrt{5}}
\left(\begin{array}{cccc}
1 & 1 & 2 & 2 \\
\xi & -1-\xi & 2\xi & 2(-1-\xi)\\
0 & 0 & 2-\xi & 2-(-1-\xi) \\
0 & 0 & \xi(2-\xi) & (-1-\xi)(2-(-1-\xi))  
\end{array}
\right)
\]
where 
\[
D=
\begin{pmatrix}
2 & -1 \\
-1 & 3
\end{pmatrix}
\]
and $DM$ is the matrix of embeddings of the ideal $\pf=(2-\xi)$.

\end{exmp}

The constructions from $\QQ(\zeta_{p}+\zeta_{p}^{-1})$ are particularly useful for coding applications to fading channels, as will be discussed in next section.
We conclude this section by mentioning two other families of subfields of $\QQ(\zeta_{p^r}+\zeta_{p^r}^{-1})$ that can be used to obtain lattices from codes.

Since the maximal totally real subfield $K^+=\QQ(\zeta_{p^r}+\zeta_{p^r})$ has degree $p^{r-1}(p-1)/2$, its Galois group contains a subgroup of order $p^{r-1}$, namely its unique Sylow $p$-subgroup, which in itself contains a cyclic subgroup $P$ of order $p$. Let $K^P$ be the subfield fixed by $P$, which has degree $p^{r-2}(p-1)/2$ over $\QQ$. Let $\pf$ be the prime in $K^P$ above $p$. As before,
\[
p^{r-1}\frac{(p-1)}{2}=e(\Pf|p)=e(\Pf|\pf)e(\pf|p)
\]
and thus $\pf$ is totally ramified. Let $C \subset \FF_p^N$ be a $k$-dimensional code with 
$C\subset C^\perp$. Then the field $K^P$ enables the construction of lattices of rank 
$Np^{r-2}(p-1)/2$, with respect to the bilinear form $\langle x,y \rangle = \sum_{i=1}^N \tr_{K^P/\QQ}(x_iy_i/p)$.

When $r=2$, $K^+=\QQ(\zeta_{p^2}+\zeta_{p^2}^{-1})$ contains a subextension of degree $p$ over $\QQ$, which again, as above, provides lattices of rank $N\frac{p-1}{2}$ from a code $C$ with $C\subset C^\perp$.

\begin{exmp}\rm
Take $p=5$, $r=2$, then $K=\QQ(\zeta_{25})$ contains the subfield of degree 5, given by the minimal polynomial
\[
p(X)=X^5-10X^3+5X^2+10X+1.
\]
\end{exmp}

%
%
\section{Applications to Coset Coding}
\label{sec:wc}

\revised{
We next discuss applications of the above constructions to coset coding for block fading channels. The particular case of wiretap block fading channels is discussed in Subsection \ref{subsec:wiretap}.} In what follows, we keep the notation consistent with those of the previous section.

\subsection{Coset Encoding for Lattices}
\label{subsec:enclatt}

Consider two nested lattices $\Lambda_e\subset\Lambda_b \subset \RR^{nN}$. We partition the lattice $\Lambda_b$ into a union of
disjoint cosets of the form $\Lambda_e+\cv$, where $\cv$ is an $nN$-dimensional vector. Suppose that we have $q^k$ cosets:
\[
\Lambda_b = \bigcup_{j=1}^{q^k} (\Lambda_e+\cv_j),
\]
and consider the mapping of $\sv\in\FF_q^k$ to a coset:
\[
\sv\mapsto \Lambda_e+\cv_{j\left(\sv\right)}.
\]
\revised{A point $\xv \in \Lambda_b$ can thus be written as $\xv\in \Lambda_e+\cv_{j\left(\sv\right)}$ or equivalently we write that $\xv\in\Lambda_b$} is of the form
\begin{equation}\label{eq:xsent}
\xv = \rv + \cv \in \Lambda_e + \cv.
\end{equation}
\revised{
This gives an explicit labeling of the lattice point $\xv \in \Lambda_b$, namely $k$ symbols in $\FF_q$ are needed to label $\cv$, and the rest of the symbols available are used to label points in $\Lambda_e$. 
}

Construction A from Example \ref{ex:a} can be used for lattice coset encoding 
as follows (with $n=1$). 
Let $C$ be an $(N,k)$ linear binary code. If $\Lambda_b=\rho^{-1}(C)/\sqrt{2}$ and $\Lambda_e=(2\ZZ)^N/\sqrt{2}$, we can partition $\Lambda_b$ as 
\[
\Lambda_b = \frac{1}{\sqrt{2}}\left( (2\ZZ)^N+C \right)=\frac{1}{\sqrt{2}}\bigcup_{\cv_i\in C}((2\ZZ)^N+\cv_i).
\]
\revised{This illustrates the benefits of coset encoding: the labeling of $\Lambda_e$ is typically difficult to perform while it is easy to label $\ZZ^N$ and encode a codeword $\cv_i$ using $k$ symbols and its generator matrix.}

\subsection{Coset Coding for Block Fading Channels}  
\label{subsec:block}

We now focus on the case where the communication channel is a block fading channel, and Perfect Channel State Information (CSI) is assumed at the receiver (Bob). 
With the aid of an in-phase/quadrature component interleaver, it is possible to remove the phase of the complex fading coefficients to obtain a real fading which is Rayleigh distributed and guarantee that the fading coefficients are 
independent from one real symbol to the next \cite[sec~2.1]{Oggier-2}.
Let $N$ be the coherence time of the block fading.
This is modeled by
\begin{equation}\label{eq:blockfading}
\begin{array}{ccl}
Y & = & \diag(\hv_{b})X+V_{b}
\end{array}
\end{equation}
where the $n\times N$ matrix $X$ is the transmitted signal, and the $n\times N$ matrix $V_{b}$ is the Gaussian noise at Bob. By channel assumption, $V_{b}$ has zero mean and variance $\sigma_{b}^{2}$. The fading matrix is given explicitly by
\begin{align*}
\diag(\hv_{b})&=\diag(|h_{b,1}|,\ldots,|h_{b,n}|)
\end{align*}
where the fading coefficients 
$h_{b,i}$ are complex Gaussian random variables with variance
$\sigma_{h,b}^{2}$, so that $|h_{b,i}|$
are Rayleigh distributed with parameter $\sigma_{h,b}^{2}$, for all $i=1,\ldots,n$.

One may vectorize the channel model (\ref{eq:blockfading}) and obtain an $nN-$dimensional 
lattice structure from the transmitted signal, that is,
\begin{align*}
\mathsf{vec}\left(Y\right)&=\mathsf{vec}\left(\diag\left(\hv_{b}\right)X\right)+\mathsf{vec}\left(V_{b}\right)\\
&=
\left(\!\!\!
\begin{array}{ccc}
\diag(\hv_b)&&\\
&\ddots& \\
&& \diag(\hv_b)
\end{array}\!\!\!
\right)\mathsf{vec}(X)+\mathsf{vec}(V_b),
\end{align*}
We now interpret the $n\times N$ dimensional codeword $X$ as coming from a lattice by writing
\[
\mathsf{vec}(X)=M_C\uv
\]
where $\uv\in\ZZ^{nN}$ and $M_C$ denotes the $nN\times nN$ generator matrix of the lattice $\Lambda_b$. This representation allows us to consider the transmitted signals as real lattice points, and coset coding simplifies to the procedure for lattices discussed in Subsection \ref{subsec:enclatt}.

\subsection{Wiretap Encoding of Algebraic Lattices}  
\label{subsec:algenc}

Let $K$ be a totally real number field of degree $n$ with ring of integers $\mathcal{O}_K$ and $n$ embeddings $\sigma_{1},\sigma_{2},\ldots,\sigma_{n}$ of $K$ into $\RR$. 
As before, suppose that $K$ contains a prime $\pf$ above $p$ which totally ramifies. 
Let $C \subset C^{\perp}$ be a linear $(N,k)$ code, with a generator matrix
\[
G=
\begin{pmatrix}
I_k & A \\
\end{pmatrix}.
\]
Recall that $x_j\in\Oc_K$ is written as $x_j=\sum_{l=1}^nx_{jl}\nu_l$ in an integral basis 
$\{\nu_1,\ldots,\nu_n\}$, and that
\[
\sigma(x_j)=(\sigma_1(x_j),\ldots,\sigma_n(x_{j}))
\]
is the canonical embedding of $K$.

A lattice point $\xv$ in $\Gamma_C$ is, as shown in (\ref{eq:x}), given by
\[
\xv=(\sigma(x_1),\ldots,\sigma(x_k),\sum_{j=1}^ka_{j,1}\sigma(x_j)+\sigma(x'_{k+1}),\ldots,
\sum_{j=1}^ka_{j,N-k}\sigma(x_j)+\sigma(x'_N))
\]
with $x'_{k+1},\ldots,x'_N \in \pf$,
which can be rearranged into an $n\times N$ matrix $X$ as follows:
\[
X=
\begin{pmatrix}
\sigma(x_1)^T,\ldots,\sigma(x_k)^T,(\sum_{j=1}^ka_{j,1}\sigma(x_j)+\sigma(x'_{k+1}))^T,\ldots,
(\sum_{j=1}^ka_{j,N-k}\sigma(x_j)+\sigma(x'_N))^T
\end{pmatrix},
\]
that is (with normalization)
\[
X=
\frac{1}{\sqrt{p}}
\begin{pmatrix}
\sigma_1(x_{1}) & \ldots & \sigma_1(\sum_{j=1}^ka_{j,1} x_{j}+x'_{k+1}) & \ldots & \sigma_1(\sum_{j=1}^ka_{j,N-k} x_{j}+x'_{N})\\
\sigma_2(x_{1}) & \ldots &  \sigma_2(\sum_{j=1}^ka_{j,1} x_{j}+x'_{k+1}) & \ldots & \sigma_2(\sum_{j=1}^ka_{j,N-k} x_{j}+x'_{N})\\
\vdots & & \vdots & & \vdots \\
\sigma_n(x_{1}) & \ldots & \sigma_n(\sum_{j=1}^ka_{j,1} x_{j}+x'_{k+1}) & \ldots & \sigma_n(\sum_{j=1}^ka_{j,N-k} x_{j}+x'_{N})\\
\end{pmatrix}.
\]
Now, coset encoding is performed by setting $\Lambda_b=\rho^{-1}(C)/\sqrt{p}=\Gamma_C$ and $\Lambda_e=\pf^N/\sqrt{p}$, using the fact that
\[
\frac{\rho^{-1}(C)}{\sqrt{p}} = \frac{1}{\sqrt{p}} \left( \pf^N + C \right) 
=\frac{1}{\sqrt{p}} \bigcup_{c_i\in C} \left( \pf^N + c_i \right). 
\]
In other words, the choice of $(x_1,\ldots,x_k)\in\Oc_K^k$ determines a coset of $\pf^N$, since 
\[
(\rho(x_1),\ldots,\rho(x_k))=(\pf+c_{i1},\ldots,\pf+c_{ik})
\]
and consequently
\begin{align*}
&(\rho(x_1),\ldots,\rho(x_k),\sum_{j=1}^k a_{j,1}\rho(x_j),\ldots,\sum_{j=1}^k a_{j,N-k}\rho(x_j)) \\
&=(\pf+c_{i1},\ldots,\pf+c_{ik},\ldots \pf+c_{in}) \\
&=\pf^N+c_i,
\end{align*}
for $c_i$ a codeword in $C$.
This explains why the construction given in Section 2 is useful for fading coset codes: it gives a way to perform coset encoding of lattices obtained from number fields, and in particular from totally real number fields. 
\begin{exmp}\rm
When $K=\QQ(\zeta_p+\zeta_p^{-1})$, the lattice obtained over $\pf$, with respect 
to the trace form, is isomorphic to $\ZZ^{(p-1)/2}$ \cite{Oggier-2}. Thus, the construction of Subsection 
\ref{subsec:cyc} is particularly suited for coset encoding, since $\ZZ^N$ is the easiest lattice to label.
\end{exmp}
Next, we detail why lattices from totally real number fields are good for reliability as well.

Denote the first row of $X$ by $\xv_1$ ($\sigma_1$ is the identity map): 
\[
\xv_1=(x_{1},\ldots,x_k,\sum_{j=1}^ka_{j,1} x_{j}+x'_{k+1},\sum_{j=1}^ka_{j,N-k} x_{j}+x'_{N} )
\]
and the $i$th row of $X$ is obtained as $\mathbf{x}_{i}=\sigma_{i}\left(\mathbf{x}_1\right)$: 
\begin{equation}\label{eq:vecX}
X=
\begin{pmatrix}
\xv_1 \\
\vdots\\
\xv_n
\end{pmatrix}
=
\begin{pmatrix}
x_1 & \ldots  & x_N \\
\sigma_2(x_1) & \ldots & \sigma_2(x_N) \\
\vdots &                & \vdots \\
\sigma_n(x_1) & \ldots & \sigma_n(x_N) 
\end{pmatrix}.
\end{equation}
Alternatively, the $j$th column of $X$, $j=1,\ldots,N$, can be 
seen as a lattice point
\begin{equation}\label{eq:ptx}
(\sigma_1(x_j),\ldots,\sigma_n(x_j))
\end{equation}
in the lattice $\Oc_K$ with the bilinear form given in (\ref{eq:trtr}),
or more precisely in the coset $\pf+c_{ij}$.
In this case, since $K$ is chosen to be totally real, it follows from the injectivity of $\sigma_i$ that $||\xv_i||\neq 0$ for all $i$ \cite{Oggier-2}.  
In fact, for every non-zero coefficient of the first row 
$\xv_1$, all the entries in the corresponding columns will also have non-zero coefficients. Conversely, 
each zero coefficient on the first row gives a column of zeros, and to have  
$||\xv_i||=0$ for some $i$ means that $||\xv_i||=0$ for all $i$, and so $X$ can only contain zeros. 
This property is referred to as full diversity, and it is well known \cite{Oggier-2}  that full diversity is a sufficient condition to ensure 
Bob's reliability over fading channels.

\subsection{Applications to Wiretap Channels}
\label{subsec:wiretap}

\revised{
Coset encoding is also particularly useful in the context of wiretap channels.
}
A wiretap channel, as introduced by Wyner \cite{Wyner-I}, is a broadcast channel where Alice, a legitimate sender, communicates to Bob in the presence of an eavesdropper Eve. 
The underlying assumption is that the channel between Alice and Bob is different from that between 
Alice and Eve. 
A wiretap code has the purpose of exploiting this difference in channels and providing reliability and confidentiality for the communication between Alice and Bob.
When the channel is real (or complex), wiretap codes are typically designed from lattices.

The idea behind wiretap coding is that, instead of having a one-to-one correspondence
between a vector of information and a lattice point, this vector of
information is mapped to a set of codewords, after
which the codeword to be actually transmitted is chosen randomly from this set. \revised{This can be done using coset coding, which we explain next by revisiting the generic coset encoding discussed in Subsection \ref{subsec:enclatt}.
}

Let $\sv\in\FF_q^k$ be the confidential information vector that Alice wishes to send to Bob. Consider again two nested lattices $\Lambda_e\subset\Lambda_b \subset \RR^{nN}$, where \revised{the subscript $b$ refers to the lattice used for Bob's reliability, while $e$ refers to the lattice used because of Eve.
The partition $\Lambda_b = \bigcup_{j=1}^{q^k} (\Lambda_e+\cv_j)$ into $q^k$ cosets of $\Lambda_e$ is used as follows: Alice maps the confidential message $\sv$ to a coset
\[
\sv\mapsto \Lambda_e+\cv_{j\left(\sv\right)}
\]
then she} randomly chooses a point $\xv\in \Lambda_e+\cv_{j\left(\sv\right)}$ and
sends it over the wiretap channel.
This is equivalent to choosing a random vector $\rv\in\Lambda_e$ and transmitting a lattice point $\xv\in\Lambda_b$ of the form:
\begin{equation}\label{eq:xsent}
\xv = \rv + \cv \in \Lambda_e + \cv.
\end{equation}

\revised{We illustrate how Construction A from Example \ref{ex:a} can be used for wiretap lattice encoding as follows (with $n=1$). }

\begin{exmp}\rm
Consider the 2-dimensional repetition code $C=\{(0,0),(1,1)\}$. Then
\[
\rho^{-1}(C)=(2\ZZ)^2+C=((2\ZZ^2)+(0,0))\cup((2\ZZ)^2+(1,1)).
\]
Since $C=C^\perp$, the lattice
\[
\frac{\rho^{-1}(C)}{\sqrt{2}}
\]
is a 2-dimensional unimodular lattice, thus equivalent to $\ZZ^2$.
In a wiretap encoder, Alice has $k=1$ secret bit that she uses to choose either the coset 
$(2\ZZ^2)+(0,0)$ or $(2\ZZ^2)+(1,1)$. Then, depending on how many bits of randomness she is using, 
she will pick one point at random within the chosen coset.
\end{exmp}

In a case of wiretap block fading channel, Alice wants to send data to Bob on a wiretap block fading channel where an eavesdropper Eve is trying to intercept the data through another block fading channel. Perfect Channel State Information (CSI) is now assumed at both receivers. 
Under the same assumptions for Eve's channel as that given for Bob's channel in (\ref{eq:blockfading}), we get a wiretap block fading channel given by
\[
\begin{array}{ccl}
Y & = & \diag(\hv_{b})X+V_{b},\mbox{ and}\\
Z & = & \diag(\hv_{e})X+V_{e}
\end{array}
\]
which can be vectorized to obtain
\begin{align*}
\mathsf{vec}\left(Y\right)&=\mathsf{vec}\left(\diag\left(\hv_{b}\right)X\right)+\mathsf{vec}\left(V_{b}\right)\\
&=
\left(\!\!\!
\begin{array}{ccc}
\diag(\hv_b)&&\\
&\ddots& \\
&& \diag(\hv_b)
\end{array}\!\!\!
\right)\mathsf{vec}(X)+\mathsf{vec}(V_b),
\end{align*}
and
\begin{align*}
\mathsf{vec}\left(Z\right)&=\mathsf{vec}\left(\diag\left(\hv_{e}\right)X\right)+\mathsf{vec}\left(V_{e}\right)\\
&=
\left(\!\!\!
\begin{array}{ccc}
\diag(\hv_e)&&\\
&\ddots& \\
&& \diag(\hv_e)
\end{array}\!\!\!
\right)\mathsf{vec}(X)+\mathsf{vec}(V_e).
\end{align*}
thanks to which wiretap coding simplifies to coset encoding as described above.

It was shown in \cite{icc,BO13,SS14} that the code design criterion to increase Eve's confusion is to choose $\Lambda_e \subset \Lambda_b$ such that
\begin{equation}
\sum_{\xv \in \Lambda_e \cap \mathcal{R}}
\prod_{||\xv_i||\neq 0}\frac{1}{\left\Vert \mathbf{x}_{i}\right\Vert ^{N+2}}
\label{eq:sum-block}
\end{equation}
is minimized, where $\xv_i$ is the $i$th row of some $n\times N$ matrix $X$ such that 
$\xv=\mathsf{vec}(X)$ is a point in $\Lambda_e$. 
The region $\mathcal{R}$ describes a finite constellation, \revised{that is, a finite set of lattice points used as signals,} carved out of the infinite lattice $\Lambda_e$. 

\revised{
Wiretap encoding using algebraic lattices as explained in Subsection \ref{subsec:algenc} inherits the same advantages as coset encoding in terms of diversity. More precisely, since $\Lambda_e \subset \Lambda_b$ are both lattices built over the same 
totally real number field, they both enjoy the property of full diversity. As already mentioned, diversity of $\Lambda_b$ ensures 
Bob's reliability over fading channels. While it may initially seem that diversity for $\Lambda_e$ is not a good thing, it turns out that}
having all the $n$ terms $\left\Vert \mathbf{x}_{i}\right\Vert$ non-zero \revised{minimizes} the sum (\ref{eq:sum-block}), which becomes
\begin{equation}
\sum_{\xv \in \Lambda_e \cap \mathcal{R}}
\prod_{i=1}^n\frac{1}{\left\Vert \mathbf{x}_{i}\right\Vert ^{N+2}}
=
\sum_{\xv \in \Lambda_e \cap \mathcal{R}}
\prod_{i=1}^n\frac{1}{\left\Vert \sigma_i(\mathbf{x}_{1})\right\Vert ^{N+2}}.
\label{eq:sum-div}
\end{equation}
\revised{
Diversity thus appears to increase Eve's confusion as well. This can be intuitively explained as follows: ideally Eve should be able to decode perfectly random symbols, up to her channel capacity, which will in turn protect the confidential symbols. 
}

%
%

\section{Conclusion}
In this paper, we proposed  a construction of lattices over number fields using linear codes over finite fields, mimicking Construction A from binary codes. We studied into details the case 
of Galois number fields, where there is a prime that totally ramifies. Maximal totally real subfield of the cyclotomic field $\QQ(\zeta_{p^r})$ are particularly considered. Applications to coset encoding for block fading channels were presented, 
\revised{giving an explicit bit labeling that guarantees the diversity of the coding scheme. We further discuss the special case of wiretap codes,}
showing why the proposed algebraic lattices are not only suited for coset encoding, but also optimize the code design criterion for wiretap codes. Future work involves a deeper study of the obtained lattices.
In particular, as shown in \cite{icc}, (\ref{eq:sum-div}) then becomes
\[
\sum_{\xv\in\Lambda_{e}\cap \Rc}\prod_{i=1}^{n}\frac{1}{\left\Vert \sigma_{i}\left(\mathbf{x}_1\right)\right\Vert ^{N+2}}=\sum_{\xv\in\Lambda_{e}\cap \Rc}\frac{1}{N_{K/\mathbb{Q}}\left(\left\Vert \mathbf{x}_1\right\Vert ^{2}\right)^{\frac{N}{2}+1}},
\]
and, among algebraic lattices, finding those which further optimize this design criterion is open.

%
%

\end{document}